\documentclass[conference]{IEEEtran}
\IEEEoverridecommandlockouts

\usepackage{cite}
\usepackage{amsmath,amssymb,amsfonts}
\usepackage{macros}
\newtheorem{remark}{Remark}
\newtheorem{assumption}{Assumption}
\newtheorem{problem}{Problem}
\newtheorem{theorem}{Theorem}

\usepackage{subcaption}
\usepackage{float}
\usepackage{empheq} 
\usepackage{mdframed} 
\usepackage{epsfig} 
\usepackage{algorithmic}
\usepackage{graphicx}
\usepackage{textcomp}
\usepackage{xcolor}
\usepackage[hidelinks]{hyperref}
\def\BibTeX{{\rm B\kern-.05em{\sc i\kern-.025em b}\kern-.08em
    T\kern-.1667em\lower.7ex\hbox{E}\kern-.125emX}}
\begin{document}

\title{Quasi-Static Control of Discrete Cosserat Rod\\
{
}
\thanks{}
}

\author{\IEEEauthorblockN{Srishti Siddharth}
\IEEEauthorblockA{\textit{Centre for Systems and Control} \\
\textit{Indian Institute of Technology Bombay}\\
Mumbai, India \\
214230005@iitb.ac.in}
\and
\IEEEauthorblockN{Domenico Campolo}
\IEEEauthorblockA{\textit{School of Mechanical and } \\
\textit{Aerospace Engineering} \\
\textit{Nanyang Technological University}\\
Singapore\\
d.campolo@ntu.edu.sg}
\and
\IEEEauthorblockN{Ravi Banavar}
\IEEEauthorblockA{\textit{Centre for Systems and Control} \\
\textit{Indian Institute of Technology Bombay}\\
Mumbai, India \\
banavar@iitb.ac.in}
}

\maketitle

\begin{abstract}
In this paper, we design feedback control laws for soft robots modelled using the Cosserat rod theory,  which is spatially discretised using the Piecewise Constant Strain (PCS) approach. The PCS approach approximates the nonlinear PDEs describing the Cosserat rod by a finite-dimensional system of nonlinear ODEs. This simplification results in a model describing soft robots which is similar to the serial rigid-link manipulators. We design feedback control laws for the quasi-static PCS model by using external  wrenches as control inputs. The control laws are designed based on feedback linearisation in strain and task spaces. An extensive set of numerical results demonstrates the performance of the control laws for end-effector trajectory tracking and shape control of soft robots.
\end{abstract}

\section{Introduction}
\label{sec:introduction}

\IEEEPARstart{I}{n} many engineering and scientific applications, geometrically exact rods are used to obtain accurate models for highly deformable slender rods. Such flexible structures find applications in continuum robots \cite{Gravagne,Cosimo_2023}, DNA strands experiencing various mechanical forces \cite{Forth_2008} and flexible microrobots for endovascular surgery \cite{Bailly_2011}. The geometrically exact rod \cite{Antman} or the Cosserat rod provides a mathematically elegant and high-fidelity model for such slender rods undergoing large deformations, unlike Euler-Bernoulli and Timoshenko beam theories, which account for small deformations.

The variational formulation of the geometrically exact rod and its numerical aspects were presented in the extensive work of Simo et al. given in \cite{Simo_1986} and \cite{Simo_1988}. 
The book by Antman \cite{Antman} presents the theory of planar and spatial geometrically exact rods. The modern treatment of geometrically exact rods is directed towards algorithm design and implementation based on the special Cosserat rod theory, which describes the motion of the rod using the positions of material points lying on its centerline and the orientations of rigid cross-sections attached at each material point. 
%

The soft robotics community has proposed various numerical models
based on the Cosserat rod theory. The Piecewise Constant Curvature (PCC) model \cite{Falk_2015,Jones_2006} assumes that a continuum robot can be represented by a series of connected circular arcs, which is not valid in the case of out-of-plane external loads. The Piecewise Constant Strain (PCS) model introduced by Renda et al. \cite{Renda_2016,Renda_2017,Renda_2018}, assumes that a soft robot consists of multiple sections, each with constant linear and angular strains. The resulting geometric structure of the model is similar to the traditional rigid-link manipulators, thus unifying the theory of rigid-link and soft robots \cite{park_lynch,Murray}. This allows us to extend the control design theory developed for rigid-link manipulators to soft robots. Other strain based models include the Geometric Variable Strain (GVS) model \cite{Boyer_Renda_2021} and the Piecewise Linear Strain (PLS) model \cite{Li_2023}, which include discrete statics and dynamics models of the Cosserat rod based on strain basis functions.

Several model-based controllers have been proposed for soft robots, which are inherently underactuated. 
Lyapunov function-based control laws have been developed in \cite{Apoorva_2011} and \cite{Molu_Chen} for dynamics of soft robots. An inverse kinematics-based controller for the GVS model has been developed in \cite{RendaIK}. Finally, a robust control law for the PLS static model has been proposed in \cite{Li_Renda}. For the sake of conciseness, we do not provide an extensive literature survey on control of soft robots. In this paper, we will focus on design of control laws for the quasi-static PCS model. Thus, the end-effector trajectory tracking and shape control tasks will be achieved by considering slow motion of the robot to satisfy the conditions for quasi-static operation. 
%

\textit{Contributions:} The contribution of the paper is to primarily develop the model-based control theory of soft robots. This is in contrast to the modern development of soft robots which focuses on learning-based control strategies. We formulate the inverse kinematics for soft robots, which provides various strains for a given end-effector pose of the robot. We then derive simple and effective control schemes based on feedback linearisation in strain and task spaces, while remaining within the framework of the Cosserat rod theory. We present three numerical experiments to validate the control approaches.

\textit{Outline:} In Section \ref{sec:continuos_time_statics}, we present a brief review of the statics of the Cosserat rod. In Section \ref{sec:PCS}, we recall the theory of the PCS static model. In Section \ref{sec:PCS_IK}, we propose the inverse kinematics algorithm, followed by the design of feedback control laws in Section \ref{sec:control}. In Section \ref{sec:numerics}, we present and discuss three numerical experiments to validate our control design approaches. Section \ref{sec:conclusion} includes concluding remarks and future work.
\section{Statics of the Cosserat rod}
\label{sec:continuos_time_statics}
\subsection{Kinematics of the Rod}
The Cosserat rod theory accounts for stretching, shearing, bending and twisting strains of a rod. Let $\Linit$ be the length of the undeformed rod and let $\arclength \in [0,\Linit]$ be the arc-length parameter which is used to identify the material points on the centerline of the rod. The configuration of the rod is described with respect to an inertial frame $\{ e_1,e_2,e_3\}$, where $e_1,e_2,e_3$ are the standard unit vectors in $\R^3$. The undeformed or reference configuration of the rod is assumed to be the one in which the rod is lying in a straight configuration along the $e_1$ axis. The configuration of the rod at any time instant is completely described by the position of each material point lying on the centerline, denoted by $\x(\arclength,t) \in \R^3$ and the orientation of the cross-section attached to each point, denoted by $\rot(\arclength,t) \in \SOthree$. Thus, the configuration of the rod at a time instant $t$, defined by the map $\g: [0,\Linit] \to SE(3)$, is given as
\begin{equation}
    \g(\arclength,t) = \begin{bmatrix}
\rot(\arclength,t) & \x(\arclength,t) \\ 
0^T & 1
    \end{bmatrix}
    \label{eq:config}
\end{equation}
for all $\arclength \in [0,\Linit]$. In addition to the configuration described in \eqref{eq:config}, the kinematics of the rod is described by its strains and velocities. The velocity twist, expressed in the body-fixed frame (or the material frame) of the rod is given by
\begin{equation}
    \widehat{\velocity}(\arclength,t) = \g^{-1}(\arclength,t) \frac{\partial \g}{\partial t}(\arclength,t) = \g^{-1}(\arclength,t) \Dot{\g}(\arclength,t), 
\end{equation}
where
\begin{equation}
    \velocity (\arclength,t) = \begin{bmatrix}
        \angvel^T (\arclength,t) & \linvel^T (\arclength,t)
    \end{bmatrix}^T \in \R^6,
    \label{eq:velocity}
\end{equation}
$\widehat{(\cdot)}:\R^6 \to \sethree$ is the hat map, 
$\linvel (\arclength,t) \in \R^3$ and $\angvel (\arclength,t) \in \R^3$ are the linear and angular velocities, respectively. 

Similarly, the strain twist, expressed in the material frame, is given by 
\begin{equation}
    \widehat{\strain} (\arclength,t) = \g^{-1}(\arclength,t) \frac{\partial \g}{\partial \arclength}(\arclength,t)  = \g^{-1} (\arclength,t) \g^{\prime} (\arclength,t),
\end{equation}
where
\begin{equation}
    \strain (\arclength,t) = \begin{bmatrix}
        \angstr^T (\arclength,t) & \linstr^T (\arclength,t)
    \end{bmatrix}^T \in \R^6,
    \label{eq:strain}
\end{equation}
$\linstr (\arclength,t) \in \R^3$ and $\angstr (\arclength,t) \in \R^3$ are the linear and angular strains, respectively. The first component of $\linstr (\arclength,t)$ represents the stretching strain along $e_1$, and the remaining components represent the shearing strains along $e_2$ and $e_3$ axes, respectively. Similarly, the first component of $\angstr (\arclength,t)$ represents the twisting strain about the $e_1$ axis, and the remaining components represent bending strains about $e_2$ and $e_3$ axes, respectively. 
\begin{figure}[H]
    \centering
    \includegraphics[width=7cm,height=7cm,keepaspectratio]{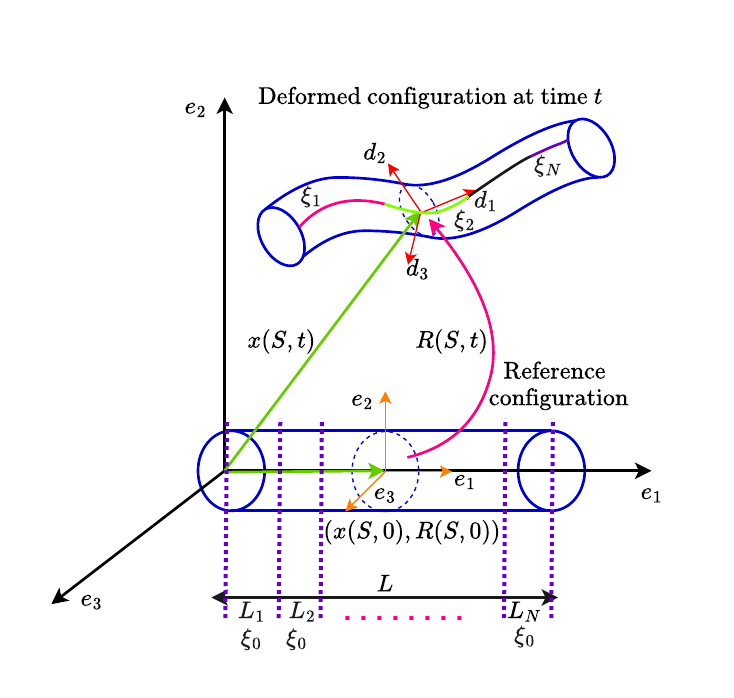}
    \caption{Deformation of the three-dimensional Cosserat rod discretised by the PCS approach.}
    \label{fig:cosserat_rod}
\end{figure}
\begin{remark}
    For the sake of concise notation, the space and time dependence of the quantities described in \eqref{eq:config} - \eqref{eq:strain} will be sometimes ignored. 
\end{remark}
The relationship between the velocity in \eqref{eq:velocity} and the strain in \eqref{eq:strain} is given by
\begin{equation}
    \velocity^{\prime} = \Dot{\strain} - \ad_{\strain} \velocity
    \label{eq:velevolution}
\end{equation}
due to the equality of the mixed partial derivatives, i.e., $(\Dot{\g})^{\prime} = \Dot{(\g^{\prime})}$. Here, $\ad_{\strain}$ is the adjoint map for the strain twist vector described in \cite{Renda_2018}.
\subsection{Statics of the rod}
The strain potential energy is stored in the rod due to various modes of deformation. The strain potential energy of the rod at time instant $t$ is given as
\begin{equation}
    \mathcal{U}(t) = \int_{0}^{\Linit} \frac{1}{2} (\strain (\arclength) - \strain_{0})^T \Sigma (\arclength)(\strain (\arclength) - \strain_{0}) \, d \arclength, 
    \label{eq:contn_PE}
\end{equation}
where $\strain_{0}= \begin{bmatrix}
    0 & 0 & 0 & 1 & 0 & 0
\end{bmatrix}^T$ is the strain for the reference configuration and $\Sigma(\arclength) =\text{diag}(GJ_{x}(\arclength), EI_{y}(\arclength), EI_{z}(\arclength), EA(\arclength), GA(\arclength), GA(\arclength))$ is the stiffness matrix of the rod, $E$ is the Young's modulus and $G$ is the shear modulus of the rod. The relation between the moduli is given by $G = \frac{E}{2(1+\nu)}$
where $\nu$ is the Poisson's ratio of the rod \cite{Antman}.

The cross-sectional area of the rod at $\arclength$ is denoted by $A(\arclength)$ and the second moment of inertia is given by $\secondmoi(\arclength) = \text{diag}(J_x(\arclength), J_y(\arclength), J_z(\arclength))$. For a circular cross-section, we have $\secondmoi(\arclength) = \text{diag}(2,1,1)\frac{A^2(\arclength)}{4\pi}$.
Assuming a linear visco-elastic constitutive relation, the material stress, also called the wrench of internal elastic forces, is given as
\begin{equation}
    F_i(\arclength,t) = \stiffness(\arclength) (\strain(\arclength,t)-\strain_0) + \damping \Dot{\strain}(\arclength,t),  
\end{equation}
where $\damping = \text{diag}(J_x,3J_y,3J_z,3A,A,A)\shearv$ and $\shearv$ is the shear viscosity modulus. The internal elastic wrench can be written as $F_i = \begin{bmatrix}
    M_i^T & N_i^T
\end{bmatrix}^T$, where $M_i,N_i \in \R^3$ are the stresses, expressed in the material frame, due to angular and linear strains, respectively. Let $F_e = \begin{bmatrix}
    m^T_e & n^T_e
\end{bmatrix}^T$ be the external distributed wrench on the rod, where $m_e,n_e \in \R^3$ are the external distributed moment and force, respectively, expressed in the material frame. Thus, the quasi-static equation of motion of the rod, derived by applying the Lagrange-D'Alembert principle, is given as
\begin{equation}
    F_i^{\prime} + \ad_{\strain}^{*} F_i + F_e = 0,
    \label{eq:contn_statics}
\end{equation}
where the coadjoint map $\ad_{\strain}^{*}$ for the strain twist vector is defined in \cite{Renda_2018}.
In this paper, we consider the external distributed load due to gravity. Thus, the distributed wrench due to gravity is given as $ F_e = \mass \Ad^{-1}_{\g} \gravity$
where $\mass = \text{diag}(J_x,J_y,J_z,A,A,A) \rho$ is the screw inertia matrix, $\rho$ is the mass density of the rod, $\gravity = \begin{bmatrix}
    0&0&0&0&0&-9.81
\end{bmatrix}^T$ is the gravity twist with respect to the inertial frame and $\Ad$ is the Adjoint map for $\SEthree$ defined in \cite{Renda_2018}.
\section{Discrete model of the Static Cosserat rod: PCS approach}
\label{sec:PCS}
In this section, we will consider the Piecewise Constant Strain (PCS) approach to discretise the Cosserat rod \cite{Renda_2018}. The rod is spatially discretised into sections with constant strains, which results in analytical integration of the equations describing the continuous space-time equation of motion of the Cosserat rod. 
\subsection{Discrete kinematics of the Rod}
Let us consider that the rod of undeformed length $\Linit$ has been discretised into $\Nsec$ sections. Thus, $\arclength \in [0,\Linit]$ is divided into $\Nsec$ sections: $[0,\Linit_1], [\Linit_1,\Linit_2], ..., [\Linit_{\Nsec-1},\Linit_{\Nsec}]$, where $\Linit_{\Nsec} = \Linit$. According to the PCS approach, each section is assumed to have a constant strain, denoted by $\strain_\n$ for $\n = 1,2,..., \Nsec$. The spatial evolution of the configuration of the rod at time instant $t$ is described by
\begin{equation}
    \g^{\prime} = \g \widehat{\strain}.
    \label{eq:strain_evolution}
\end{equation}
Based on the PCS approach, the solution of \eqref{eq:strain_evolution} for $\arclength \in [\Linit_{\n-1}, \Linit_{\n}]$ can be written as 
\begin{equation}
    \g(\arclength) = \g(\Linit_{\n-1}) e^{(\arclength-L_{\n-1})\widehat{\strain}_n} = \g(\Linit_{\n-1}) \g_n(\arclength),
    \label{eq:g_exp}
\end{equation}
where the closed-form expression for $\g_n(\arclength) := e^{(\arclength-L_{\n-1})\widehat{\strain}_n}$ is given in \cite{Renda_2018}. 

%

The pose at $\arclength \in [\Linit_{n-1},\Linit_\n]$ can be found recursively from \eqref{eq:g_exp}, which gives us the product of exponentials (PoE) formula for the PCS model as
\begin{equation}
    \g(\arclength) = \g(0) e^{\Linit_1 \widehat{\strain}_1}e^{(\Linit_2-\Linit_1) \widehat{\strain}_2}...e^{(\arclength - \Linit_{\n-1})\widehat{\strain}_{\n}}.
    \label{eq:poe}
\end{equation}
The PoE formula \eqref{eq:poe} allows us to represent the Cosserat rod by $\strain_\n$ for $\n = 1,...,\Nsec$, which can be thought of as joint twists. Thus, the joint or the strain vector for a discrete rod is given by
\begin{equation}
    \joint = \begin{bmatrix}
        \strain^T_{1} \strain^T_{2}...\strain^T_{\Nsec} 
    \end{bmatrix}^T \in \R^{6\Nsec}.
\end{equation}
The velocity of the cross-section at $\arclength \in [0,\Linit]$ can be computed by integrating \eqref{eq:velevolution} for a section $[\Linit_{\n-1},\Linit_{\n}]$, i.e., by integrating
\begin{equation}
    \velocity^{\prime}(\arclength) = \Dot{\strain}_\n - \ad_{\strain_{\n}} \velocity(\arclength), \quad \arclength \in [\Linit_{\n-1},\Linit_{\n}].
    \label{eq:vel_at_n}
\end{equation}
Since $\strain_\n$ is constant at time instant $t$, therefore, \eqref{eq:vel_at_n} is a nonhomogeneous, linear matrix differential equation with constant coefficient, and it's solution is given by
\begin{equation}  \velocity(\arclength) = \Ad^{-1}_{\g_{\n}(\arclength)} (\velocity(\Linit_{\n-1}) + T_{\g_{\n}(\arclength)} \Dot{\strain}_{\n}),
    \label{eq:velocity_closed_form}
\end{equation}
where the closed-form expressions for $\Ad_{\g_{\n}(\arclength)}$ and $T_{\g_{\n}(\arclength)}$ are provided in \cite{Renda_2018}.
\begin{assumption}
    In this paper, we will consider a cantilever rod. Thus, the boundary conditions at the base of the rod are given by
    \begin{equation}
        \g(0,t) = I_4, \quad \velocity(0,t) = 0 \in \R^6, \quad \forall t \geq 0.
    \end{equation}
\end{assumption}
\begin{remark}
    $I_k \in \R^{k \times k}$  denotes the $k$ dimensional identity matrix.
\end{remark}
On computing $\velocity(\arclength)$ recursively from \eqref{eq:velocity_closed_form}, we get
\begin{equation}
    \velocity(\arclength) = \jacobian (\arclength,\joint) \Dot{\joint}, 
\end{equation}
where the closed-form expression for the geometric Jacobian, $\jacobian(\arclength,\joint) \in \R^{6 \times 6 \Nsec}$ is given in \cite{Renda_2018}.
\subsection{Discrete statics of the rod}
\label{sec:PCS_statics}
The discrete statics can be obtained by converting \eqref{eq:contn_statics} to the weak form and applying the piecewise constant strain assumption \cite{Renda_2018}. Thus, we have  
\begin{equation}
    D\Dot{\joint}+K (\joint-\joint^{*}) = \Bar{\jacobian}^{T}(\joint)\Bar{\mathcal{F}}_{ext}(t) + \mathcal{N}(\joint) \gravity,
    \label{eq:quasi_static_contl_eq}
\end{equation}
where $\joint^{*} = \begin{bmatrix}
    \strain^T_{0}&\strain^T_{0}&...&\strain^T_{0} 
\end{bmatrix}^T \in \R^{6\Nsec} $, $D = \text{diag}({l_1}\damping,{l_2}\damping,...,{l_n}\damping) \in \R^{6\Nsec \times 6\Nsec}$ is the generalised damping matrix, $K = \text{diag}({l_1}\stiffness,{l_2}\stiffness,...,{l_n}\stiffness) \in \R^{6\Nsec \times 6\Nsec}$ is the generalised stiffness matrix, $l_{\n}$ for $\n=1,...,\Nsec$ is the length of section $\n$, $\Bar{\jacobian}(q) = \begin{bmatrix}
    \jacobian^T(\Linit_1,q)&\jacobian^T(\Linit_2,q)&...&\jacobian^T(\Linit,q)
\end{bmatrix}^T \in \R^{6\Nsec \times 6\Nsec}$, $\Bar{\mathcal{F}}_{ext}(t)=\begin{bmatrix}
    \mathcal{F}_{ext}(\Linit_1,t)&\mathcal{F}_{ext}(\Linit_2,t)&...&\mathcal{F}_{ext}(\Linit,t)
\end{bmatrix}^T \in \R^{6\Nsec \times 1}$ and $\mathcal{N}(q) \in \R^{6\Nsec \times 6}$ is the gravitational matrix \cite{Renda_2018}. The strain potential energy of the discrete rod is given by 
\begin{align}
    \begin{split}
        \mathcal{U}_d(t)=\frac{1}{2}(\joint-\joint^{*})^TK (\joint-\joint^{*}),
    \end{split}
    \label{eq:discrete_PE}
\end{align}
which can be obtained by applying the PCS assumption to \eqref{eq:contn_PE}.

\section{Inverse kinematics for the discrete Cosserat rod}
\label{sec:PCS_IK}
The inverse kinematics (IK) problem is formulated as follows:
\begin{problem}
    Given the desired end-effector pose of the rod $\g_d(\Linit) \in \SEthree$, find the  strain vector $\joint \in \R^{6 \Nsec}$ such that the forward kinematics in \eqref{eq:poe}  is satisfied for $\arclength = \Linit$.
\end{problem}
Analogous to rigid-link serial manipulators, the inverse kinematics problem for the PCS model admits non-unique solutions owing to the presence of kinematic redundancy. A numerical solution of the IK problem is obtained by using the Newton-Raphson method. The Newton-Raphson algorithm for the PCS model, similar to the one in \cite{park_lynch} for rigid-link manipulators, is given as follows:
\begin{enumerate}
    \item Initialisation: Consider the desired end-effector pose $\g_d (\Linit)$ and an initial guess of the strain vector $\joint^{0} \in \R^{6 \Nsec}$. Set $i = 0$.
    \item Error computation: Evaluate the pose error $\widehat{\mathcal{V}} = \log(\g^{-1}(\Linit)\g_d (\Linit))$, where $\g(\Linit)$ is evaluated by substituting the strain vector $\joint^{i}$ at the $i^{th}$ iteration in \eqref{eq:poe} and $\log$ is the logarithmic map in $\SEthree$ \cite{park_lynch}. Let $\epsilon>0$ be the desired error tolerance. While $||\mathcal{V}||>\epsilon$:
    \begin{enumerate}
        \item Update the strain vector: The strain vector at the next iteration is given by
        \begin{equation}
            \joint^{i+1} = \joint^{i} + \jacobian^{\pseudoinv} (\Linit,\joint^{i}) \mathcal{V},
        \end{equation}
        where $(\cdot)^{\pseudoinv}$ denotes the Moore-Penrose pseudoinverse \cite{park_lynch}.
        \item Increment $i$.
    \end{enumerate}
\end{enumerate}
\section{Quasi-Static control of the discrete Cosserat rod}
\label{sec:control}
The shape regulation and end-effector position tracking for the rod can be achieved by formulating the control problems in the strain space or the task space. We will discuss both control approaches in this section.
\subsection{Strain space control}
\label{sec:strain_control}
Let us consider a new strain vector $\Bar{\joint} = \joint-\joint^{*}$. Thus, the discrete statics of the rod, given in \eqref{eq:quasi_static_contl_eq}, can be written as
\begin{equation}
    \Dot{\Bar{\joint}} = A \Bar{\joint} + {D^{-1}}u(\Bar{\joint},t) + {D^{-1}}\mathcal{N}(\Bar{\joint}) \gravity,
    \label{eq:statics_bar}
\end{equation}
where $A= -D^{-1} K$ and $u(\Bar{\joint},t) = \Bar{\jacobian}^{T}(\Bar{\joint})\Bar{\mathcal{F}}_{ext}(t)$.
\begin{remark}
    It should be noted that $A$ is a diagonal, Hurwitz matrix since the diagonal matrices $D^{-1}$ and $K$ contain strictly positive elements. Therefore, in the absence of any external wrenches, the rod converges to the reference configuration.
\end{remark}
\begin{theorem}
  Consider the desired strain vector of the rod, denoted by $\joint_d(t) \in \R^{6 \Nsec}$. Given a symmetric, positive definite matrix $\Bar{K} \in \R^{6\Nsec \times 6\Nsec}$ and assuming that $\Bar{\jacobian}^{T}(\Bar{\joint})$ is invertible, the feedback control law
  \begin{equation}
    \Bar{\mathcal{F}}_{ext}(t) = \Bar{\jacobian}(\Bar{q})^{-T}\bigl(D(-\Bar{K}\Bar{e} - A \Bar{e} - A \Bar{\joint}_d + \Dot{\Bar{\joint}}_d)
     - \mathcal{N}(\Bar{\joint})\gravity \bigr),
     \label{eq:FB_law_strain}
\end{equation}
  applied to \eqref{eq:statics_bar}, renders the desired trajectory locally asymptotically stable. Here, $\Bar{e}(t) = \Bar{\joint}(t) - \Bar{\joint}_d(t) $ is the strain error.
\end{theorem}
$\quad$ \textit{Proof:} From \eqref{eq:statics_bar},
the error dynamics are given by 
\begin{equation}
    \Dot{\Bar{e}} = A \Bar{e} + D^{-1}u(\Bar{\joint},t) + D^{-1}\mathcal{N}(\Bar{\joint})\gravity + A \Bar{\joint}_d - \Dot{\Bar{\joint}}_d.
    \label{eq:strain_space_err_dyn}
\end{equation}
Consider the Lyapunov function
\begin{equation}
    V(\Bar{e}) = \frac{1}{2}\Bar{e}^T \Bar{e}.
    \label{eq:Lyapunov_strain_space}
\end{equation}
Thus, $V(\Bar{e})>0$ for $\Bar{\joint} \neq \Bar{\joint}_d $ and $V(\Bar{e})=0$ for $\Bar{\joint}= \Bar{\joint}_d $. \\
On differentiating \eqref{eq:Lyapunov_strain_space} with respect to time and substituting the statics given in \eqref{eq:statics_bar} and the feedback control law in \eqref{eq:FB_law_strain}, we get
\begin{equation}
    \Dot{V}(\Bar{e}) = -\Bar{e}^T \Bar{K} {\Bar{e}}.
\end{equation}
Thus, $\Dot{V}(\Bar{e})<0$ for $\Bar{\joint} \neq \Bar{\joint}_d $ and $\Dot{V}(\Bar{e})=0$ for $\Bar{\joint}= \Bar{\joint}_d $, which ensures that the strain error converges to the origin asymptotically. 

%

\begin{remark}
    For set-point control, $\Dot{\Bar{\joint}}_d = 0$.
\end{remark}
\begin{remark}
    The control law in \eqref{eq:FB_law_strain} requires the desired strains for end-effector position tracking, which are computed using the IK algorithm given in Section \ref{sec:PCS_IK}. 
\end{remark}
\subsection{Task space control}
\label{sec:task_control}
The end-effector position tracking control can be designed by transforming the quasi-static system in \eqref{eq:quasi_static_contl_eq} from the strain space to the task space coordinates. In this section, we will consider the external tip wrench as the control input.

Let us consider the minimal coordinates $y(t) = (r(t),\x(\Linit,t)) \in \R^{6}$, where $r(t) = a(t)\theta(t)$ ($a(t) \in \R^{3}, \theta(t) \in \R$) is the exponential representation of the orientation of the cross-section at $\arclength = \Linit$ and $\x(\Linit,t) $ is the end-effector position. Then, we have
\begin{equation}
    \Dot{y} = \begin{bmatrix}
        \Dot{r} \\
        \Dot{x}(\Linit)
    \end{bmatrix} = T(r,\x(\Linit),\rot(\Linit)) \begin{bmatrix}
        \angvel(\Linit) \\
        \linvel(\Linit)
    \end{bmatrix},
    \label{eq:vel_task_space}
\end{equation}
where $T(r,\x(\Linit),\rot(\Linit)) \in \R^{6 \times 6}$ is given as
\begin{equation}
    T(r,\x(\Linit),\rot(\Linit)) = \begin{bmatrix}
        W^{-1}(r)&0\\
        0&\rot(\Linit)
    \end{bmatrix},
\end{equation}
where $W(r) = I_3 - \frac{1-\cos(||r||)}{||r||^2}\Tilde{r} + \frac{||r||-\sin(||r||)}{||r||^3}\Tilde{r}^2$ and $\Tilde{(\cdot)}:\R^3 \to \sothree$ is the hat map \cite{park_lynch} .

On substituting $\velocity = \jacobian \Dot{\joint}$ and \eqref{eq:statics_bar} in \eqref{eq:vel_task_space}, we get
\begin{equation}
    \Dot{y} = \Bar{A}(\Bar{\joint}) \Bar{\joint} + \Bar{B}(\Bar{\joint})  \mathcal{F}_{ext}(\Linit,t) + \Bar{C}(\Bar{\joint})\gravity,
    \label{eq:task_space_eq_of_motion}
\end{equation}
where $\Bar{A}(\Bar{\joint}) = T(r,\x(\Linit),\rot(\Linit)) \jacobian(\Linit,\Bar{\joint})A \in \R^{6 \times 6\Nsec}$, $\Bar{B}(\Bar{\joint})= T(r,\x(\Linit),\rot(\Linit)) \jacobian(\Linit,\Bar{\joint}) D^{-1} \jacobian^T(\Linit,\Bar{\joint}) \in \R^{6 \times 6}$ and $\Bar{C}(\Bar{\joint}) = T(r,\x(\Linit),\rot(\Linit)) \jacobian(\Linit,\Bar{\joint})D^{-1} \mathcal{N}(\Bar{\joint}) \in \R^{6 \times 6}$. The position of the tip of the rod can be written as $\x(L,t) = Py(t)$, where $P = \begin{bmatrix}
    0_{3\times 3}&I_3
\end{bmatrix}$ is the projection operator. Thus, we have
\begin{equation}
    \Dot{\x}(L) = P\Bar{A}(\Bar{\joint}) \Bar{\joint} + P\Bar{B}(\Bar{\joint})  \mathcal{F}_{ext}(\Linit,t) + P\Bar{C}(\Bar{\joint})\gravity.
    \label{eq:task_space_postion_dynamics}
\end{equation}
\begin{theorem}
    Consider a time-varying desired trajectory of the position of the tip of the rod, denoted by $x_d(\Linit,t) \in \R^3$. Given a symmetric, positive definite matrix $\Bar{K}_{t}  \in \R^{3 \times 3}$ and assuming that $P\Bar{B}(\Bar{\joint})$ has full row rank, the feedback control law 
    \begin{equation}
    \begin{aligned}
         \mathcal{F}_{ext}(\Linit,t) = (P\Bar{B}(\Bar{\joint}))^{\pseudoinv} (&-P\Bar{A}(\Bar{\joint}) \Bar{\joint}-P\Bar{C}(\Bar{\joint})\gravity\\&+\Dot{\x}_d(\Linit)-\Bar{K}_{t}e),
    \end{aligned}
\label{eq:FB_law_task}
\end{equation}
applied to \eqref{eq:task_space_postion_dynamics}, renders the desired trajectory locally asymptotically stable. Here, $e(t) = \x(L,t) - x_d(\Linit,t)$ is the position error of the tip of the rod. 
\end{theorem}
$\quad$ \textit{Proof:} From \eqref{eq:task_space_postion_dynamics}, the error dynamics are given by

\begin{equation}
    \Dot{e} = P\Bar{A}(\Bar{\joint}) \Bar{\joint} + P\Bar{B}(\Bar{\joint})  \mathcal{F}_{ext}(\Linit,t)+ P\Bar{C}(\Bar{\joint})\gravity - \Dot{\x}_d.
    \label{eq:error_dyn_task_space}
\end{equation}
Consider a Lyapunov function
\begin{equation}
    V(e) = \frac{1}{2}e^T e.
    \label{eq:Lyapunov_task_space}
\end{equation}
Thus, $V(e)>0$ for $\x(\Linit) \neq \x_d(\Linit)$ and $V(e)=0$ for $\x(\Linit) = \x_d(\Linit)$.\\
On differentiating \eqref{eq:Lyapunov_task_space} with respect to time and substituting \eqref{eq:error_dyn_task_space} and \eqref{eq:FB_law_task} in the resulting expression, we get
\begin{equation}
    \Dot{V}(e) = -e^T \Bar{K}_t e.
\end{equation}
Thus, $\Dot{V}(e)<0$ for $\x(\Linit) \neq \x_d(\Linit)$ and $\Dot{V}(e)=0$ for $\x(\Linit) = \x_d(\Linit)$, which ensures that the tip position error converges to the origin asymptotically.
\begin{remark}
    For set-point control, $\Dot{\x}_d (\Linit) = 0$.
\end{remark}

\section{Numerical results}
\label{sec:numerics}
In this section, we present three numerical experiments to validate the IK algorithm and the control laws for the discrete Cosserat rod. We consider a cantilever cylindrical rod of length $0.3$ m and radius $10^{-2} $m. The material properties of the rod are as follows: $\rho= 10^{3}$ $\text{kg}/\text{m}^3,$  $E = 10^{6}$ Pa, $\nu = 0.5$ and $\shearv = 10^2$ Pa s. The rod initially lies along the $e_1$ axis in a straight configuration. Each section of the discrete Cosserat rod modelled by the PCS approach has been spatially discretised into 40 intervals to generate the shape of the rod. 
\subsection{Inverse kinematics for the discrete Cosserat rod}
\label{sec:IK_numerics}
Let us consider the PCS model with two sections and the desired end-effector pose
\begin{equation}
    \g_d(\Linit) = \begin{bmatrix}
        \cos(\pi/4)&-\sin(\pi/4)&0&0.25\\
        \sin(\pi/4)&\cos(\pi/4)&0&0.2\\
        0&0&1&0\\
        0&0&0&1
        
    \end{bmatrix}.
\end{equation}
The IK algorithm in Section \ref{sec:PCS_IK} determines the strain vector $\joint \in \R^{12}$ such that the forward kinematics in \eqref{eq:poe} is satisfied for $\arclength = \Linit$. We consider the initial guesses 
\begin{equation}
    \joint^0_1 = \joint^{*}
\end{equation}
and 
\begin{equation}
    \joint^0_2 = \begin{bmatrix}
0&10&0&1&0&0&\strain_0^T
    \end{bmatrix}^T.
\end{equation}
Figure \ref{fig:IK_solutions} shows the configurations of the rod obtained from the IK algorithm based on $\joint^0_1$ and $\joint^0_2$. The rod admits multiple feasible configurations due to its inherent redundancy, leading to distinct solutions when different initial guesses are provided to the IK algorithm based on the Newton-Raphson method.  The strain potential energies, given in \eqref{eq:discrete_PE}, corresponding to $q^0_1$ and $q^0_2$ are $10.59$ J and $32.65$ J, respectively. 
\begin{figure}[H]
    \centering
    \includegraphics[width=0.5\linewidth]{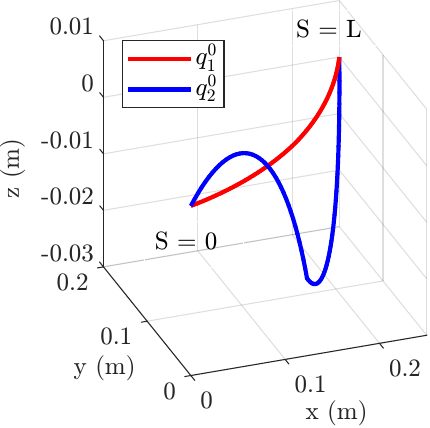}
    \caption{IK solutions for two initial guesses $q^0_1$ and $q^0_2$. }
    \label{fig:IK_solutions}
\end{figure}
\subsection{Shape regulation of the discrete Cosserat rod}
\label{sec:shape_numerics}
We consider the PCS model with two sections and the desired strain
\begin{equation}
\joint_d = \begin{bmatrix}
\joint_{d1}^T&\joint_{d2}^T
\end{bmatrix}^{T},
\end{equation}
where $\joint_{d1}= \ \begin{bmatrix}
    0&-5&0&1&0&0
\end{bmatrix}^T$ and $\joint_{d2}= \ \begin{bmatrix}
    0&10&0&1&0&0
\end{bmatrix}^T$. The desired strain vector has been chosen to generate a symmetrical shape of the rod including two sections of constant curvatures. The desired tip pose can be determined by the forward kinematics in \eqref{eq:poe}. The strain space control in \eqref{eq:FB_law_strain} with $\Bar{K} = 2 I_{12}$ has been implemented to regulate the shape of the rod. 
The closed-loop dynamics in \eqref{eq:statics_bar} have been simulated using the fourth-order Runge-Kutta method (RK4) for $5$ seconds with a time step of $10^{-3}$ seconds. 
Figure \ref{fig:a} and Figure \ref{fig:b} demonstrate the shape of the rod at various time instants and the end-effector external wrench, respectively. The external wrench applied at $\Linit_1$ has not been shown due to space limitations. It can be observed from Figure \ref{fig:a} and Figure \ref{fig:b} that the rod achieves the desired shape after $2.5$ seconds.
\begin{figure}[H]
   \begin{subfigure}{0.45\columnwidth}
    \centering
    \includegraphics[width=\linewidth]{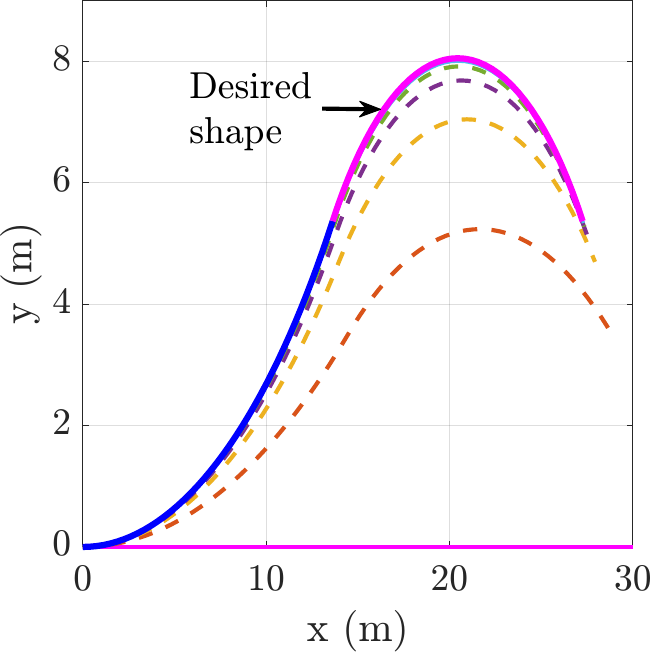}
    \caption{}
    \label{fig:a}
\end{subfigure}
\hfill
\begin{subfigure}{0.45\columnwidth}
    \centering
    \includegraphics[width=\linewidth]{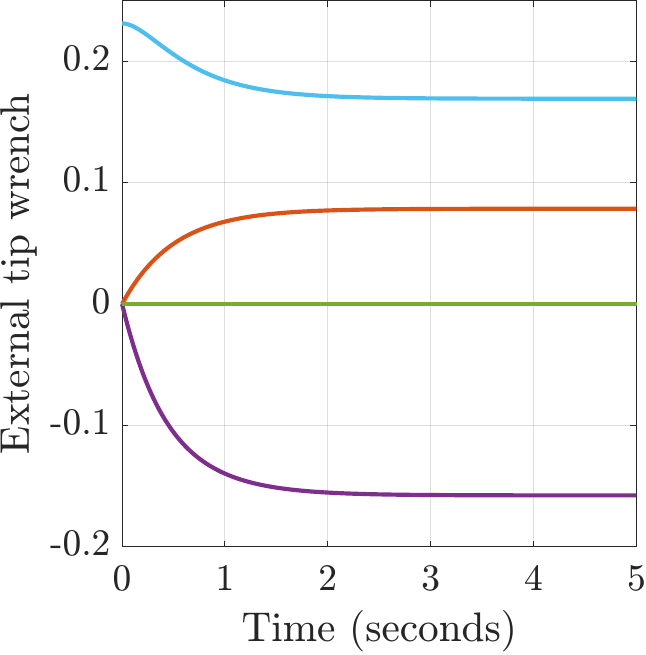}
    \caption{}
    \label{fig:b}
\end{subfigure}

    \caption{(a) Evolution of the shape of the rod for the following time instants: $t = 0$ seconds (pink), $t = 0.5$ seconds (orange), $t = 1$ second (yellow), $t = 1.5$ seconds (violet) and $t = 2$ seconds (green). (b) External tip wrench denoted by $\mathcal{F}_{ext}(\Linit) = \begin{bmatrix}
        m(\Linit)^T&n(\Linit)^T
    \end{bmatrix}^T$, where $m(\Linit) = \begin{bmatrix}
        0&m_y&0
    \end{bmatrix}^T$ Nm and $n(\Linit) = \begin{bmatrix}
        n_x&n_y&n_z
    \end{bmatrix}^T$ N. Here, $m_y,n_x,n_y$ and $n_z$ are shown in red, violet, green and blue, respectively.  }
    \label{fig:case2_shape_reg}
\end{figure}
\subsection{End-effector position tracking control of the discrete Cosserat rod}
\label{sec:shape_numerics}
We consider the PCS model with 4 sections. The objective is to ensure that the tip of the rod tracks a circular trajectory of radius $0.1$ m in the $e_2–e_3$ plane. The strain space and task space control laws in \eqref{eq:FB_law_strain} and \eqref{eq:FB_law_task} have been implemented with $\Bar{K} = 2I_{60}$ and $\Bar{K}_t = 20I_3$, respectively. The closed-loop dynamics for strain and task space control approaches have been simulated using RK4 and ode45 (MATLAB), respectively, for $5$ seconds with a time step of $10^{-3}$ seconds. Figure \ref{fig:task_space_shape} shows the time evolution of the position of the rod and the circular trajectory traced by the tip of the rod obtained from the strain space control. Figure \ref{fig:aC2} and \ref{fig:bC2} demonstrate the external wrench at the tip of the rod corresponding to strain space control and the error convergence due to task space control, respectively. The error converged to the origin in 4.5 seconds due to the strain space control. It was observed that the magnitudes of the various components of the tip external wrench under task space control were lower than the strain space control approach.
\begin{figure}
    \centering
    \includegraphics[width=0.8\linewidth]{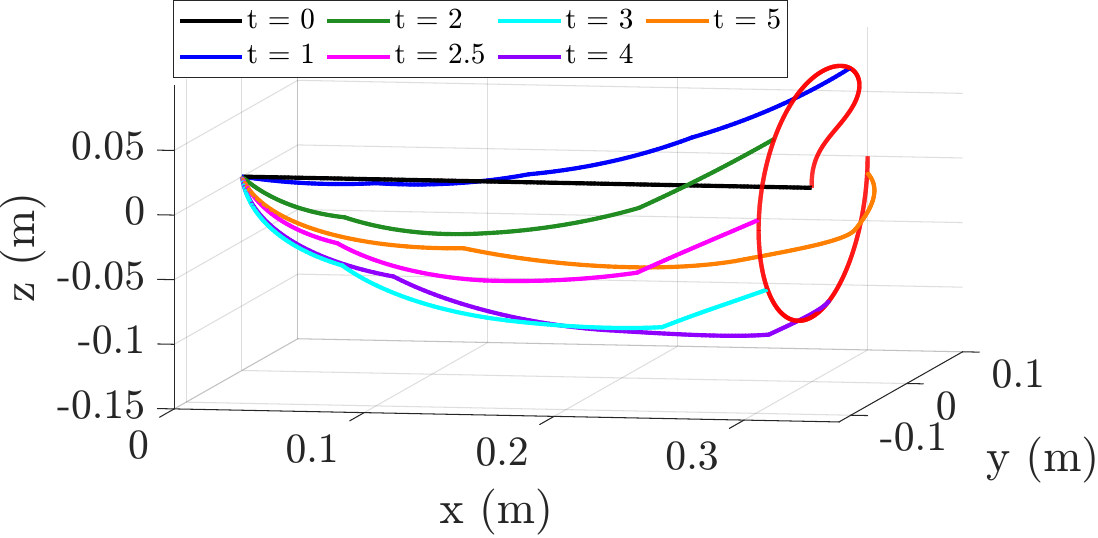}
    \caption{Time evolution of the position of the rod. The circular tip trajectory is shown in red.}
    \label{fig:task_space_shape}
\end{figure}
\begin{figure}[H]
   \begin{subfigure}{0.45\columnwidth}
    \centering
    \includegraphics[width=0.98\linewidth]{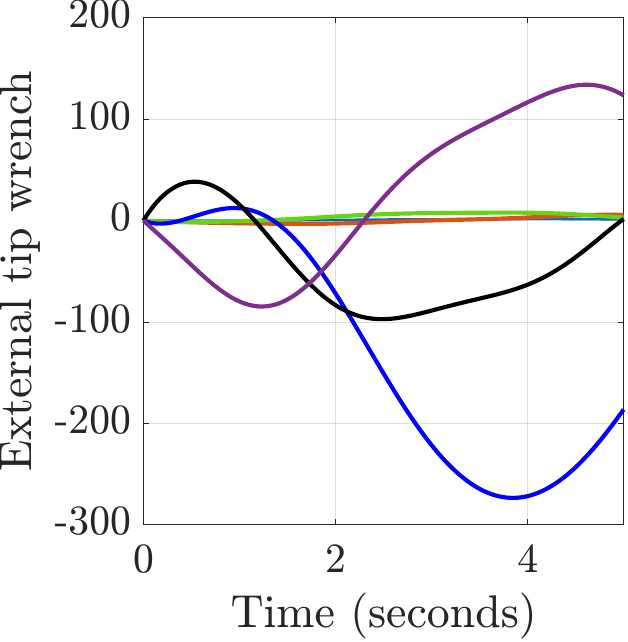}
    \caption{}
    \label{fig:aC2}
\end{subfigure}
\hfill
\begin{subfigure}{0.45\columnwidth}
    \centering
    \includegraphics[width=0.976\linewidth]{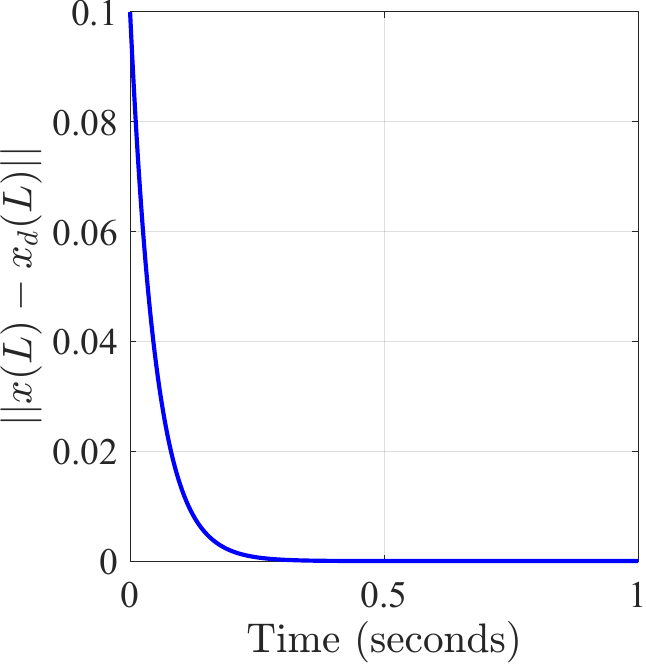}
    \caption{}
    \label{fig:bC2}
\end{subfigure}

    \caption{(a) External tip wrench obtained from the strain space control. External tip wrench is denoted by $\mathcal{F}_{ext}(\Linit) = \begin{bmatrix}
        m(\Linit)^T&n(\Linit)^T
    \end{bmatrix}^T$, where $m(\Linit) = \begin{bmatrix}
        m_x&m_y&m_z
    \end{bmatrix}^T$ Nm and $n(\Linit) = \begin{bmatrix}
        n_x&n_y&n_z
    \end{bmatrix}^T$ N. Here, $m_z,n_x, n_y, n_z$ are shown in green, blue, black and violet, respectively. (b) Error convergence due to task space control.  }
    \label{fig:case2_shape_reg}
\end{figure}

\section{Conclusions and Future Work}
\label{sec:conclusion}
In this paper, we developed feedback control laws in the strain and task spaces for the discrete Cosserat rod. The Piecewise Constant Strain (PCS) discretisation approach presents a simplified model of the Cosserat rod, which is described by a set of nonlinear PDEs evolving on $\SEthree$. The IK algorithm and the corresponding numerical test illustrated the strain redundancy of the Cosserat rod. Shape regulation of the rod was achieved by implementing the feedback control in the strain space. The performance of the proposed feedback control laws employed for end-effector position tracking of the rod was also demonstrated.

Future work will focus on the development of an IK algorithm which provides a
solution with minimum strain potential energy and an optimal control framework for cable-actuated soft robots, along with the incorporation of buckling analysis in the control design.
\section{Acknowledgments}
Srishti Siddharth's work was supported by the Prime Minister’s Research Fellowship (PMRF), funded by the Government of India. Domenico Campolo's work was supported by the Italy-Singapore collaborative DESTRO project under A*STAR grant number R22I0IR124.

\bibliographystyle{IEEEtran}
\bibliography{References2}

\end{document}